\documentclass[a4paper,UKenglish,cleveref, autoref, thm-restate, authorcolumns]{lipics-v2019}
\usepackage{algorithm}

\title{A polynomial time parallel algorithm for graph isomorphism using a quasipolynomial number of processors} %TODO Please add

\titlerunning{Parallel polynomial graph isomorphism}%optional, please use if title is longer than one line

\author{Duc Hung Pham}{Rice University, United States of America}{Hung.D.Pham@rice.edu}{https://orcid.org/0000-0002-1825-0097}{}

\author{Krishna V. Palem}{Rice University, United States of America}{Krishna.V.Palem@rice.edu}{[orcid]}{}

\author{M. V. Panduranga Rao}{Indian Institute of Technology Hyderabad, India}{mvp@iith.ac.in}{[orcid]}{}

\authorrunning{D\,Q. Public and J.\,R. Public}%TODO mandatory. First: Use abbreviated first/middle names. Second (only in severe cases): Use first author plus 'et al.'

\Copyright{Duc Hung Pham, Krishna V. Palem and M.V. Panduranga Rao}%TODO mandatory, please use full first names. LIPIcs license is "CC-BY";  http://creativecommons.org/licenses/by/3.0/

%TODO mandatory: Please choose ACM 2012 classifications from https://dl.acm.org/ccs/ccs_flat.cfm 

%TODO mandatory; please add comma-separated list of keywords

\category{}%optional, e.g. invited paper

\relatedversion{}%optional, e.g. full version hosted on arXiv, HAL, or other respository/website
\supplement{}%optional, e.g. related research data, source code, ... hosted on a repository like zenodo, figshare, GitHub, ...

\funding{This work was supported in part by a Guggenheim fellowship and by DARPA under contract number FA8750-16-2-0004.}%optional, to capture a funding statement, which applies to all authors. Please enter author specific funding statements as fifth argument of the \author macro.

\nolinenumbers %uncomment to disable line numbering

\EventEditors{John Q. Open and Joan R. Access}
\EventNoEds{2}
\EventLongTitle{42nd Conference on Very Important Topics (CVIT 2016)}
\EventShortTitle{CVIT 2016}
\EventAcronym{CVIT}
\EventYear{2016}
\EventDate{December 24--27, 2016}
\EventLocation{Little Whinging, United Kingdom}
\EventLogo{}
\SeriesVolume{42}
\ArticleNo{23}
\begin{document}

\maketitle

\begin{abstract}
Graph Isomorphism (GI) problem is a theoretically interesting problem because it has not been proven to be in P nor to be NP complete. It is known that the GI problem is in the low hierarchy of NP, and does not equal NP unless the polynomial-time hierarchy collapses to some finite level. For thirty years, the best algorithm for GI had sub-exponential running time, until Babai made a breakthrough in 2015 when announcing a quasipolynomial time algorithm for GI problem. Babai work gives the most theoretically efficient algorithm for GI, as well as a strong evidence favoring the idea that class GI $\ne$ NP and thus P $\ne$ NP. Based on Babai's algorithm, we prove that GI can further be solved by a parallel algorithm that runs in polynomial time using a quasipolynomial number of processors. We achieve that result by identifying the bottlenecks in Babai's algorithms and parallelizing them. In particular, we prove that color refinement can be computed in parallel logarithmic time using a polynomial number of processors, and the $k$-dimensional WL refinement can be computed in parallel polynomial time using a quasipolynomial number of processors. Our work suggests that Graph Isomorphism and GI-complete problems can be computed efficiently in a parallel computer, and provides insights on speeding up parallel GI programs in practice. 
\end{abstract}

\section{Introduction}
\label{sec: intro}

\subsection{Overview}
An isomorphism between two graphs $G_1 = (V_1,E_1)$ and $G_2 = (V_2, E_2)$ is a bijection $f: V_1 \rightarrow V_2 $ such that for every $e = (u,v) \in E_1$, $e' = (f(u), f(v))$ is in $E_2$. The Graph Isomorphism (GI) problem is the problem that given two graphs $G_1, G_2$, determine if there exists an isomorphism between them. 

The Graph Isomorphism problem is an interesting problem. Even though it is in NP, researchers have not successfully proven it to be in P, or to be NP-complete. In fact, it is widely believed to lie in the NP-intermediate class (NPI), the non-empty class which lies in NP but outside P in the case that P $\ne$ NP \cite{Ladner1975}.

The Graph Isomorphism problem have witnessed advancement in the practical aspect more than in the theoretical aspect. In practice, notable programs include Nauty \cite{McKay1} and Traces \cite{MCKAY2} by McKay, saucy by Darga \cite{Darga}, conauto by Presa \cite{Presa}, and bliss by Junttila \cite{Junttila}. The practical GI solvers can solve GI of random graphs of about 10000 vertices quickly \cite{McKay1}. In theory, however, the best known general theoretical bound for GI problem before Babai's work was sub-exponential time $2^{O(\sqrt{nlogn}})$. The algorithm was developed by Luks and Babai in \cite{BabaiLuks} in 1983. It took more than three decades for the next best bound, quasipolynomial, to be developed. Intuitively, Babai's work pushes GI closer to the class of P, providing more evidence suggesting that GI is not NP-complete. While the algorithm's quasipolynomial runtime is still superpolynomial, it is now remarkably closer to polynomial than to exponential time. 

On the other hand, demand for parallel algorithms is increasing. Nowadays many computational tasks are connected with big data or large input. On the other hand, more and more computers are equipped with multiple cores, and we have supercomputers with millions of cores. While memory access is still a big issue in parallel computing, it is undeniable that there is an increasingly high demand for parallel algorithms.

\subsection{Result }

In this work, we prove that we can solve the Graph Isomorphism problem in parallel polynomial time using a quasipolynomial number of processors. Using Babai's algorithm as the basis of parallelization, we demonstrate that Babai's algorithm, as well as GI algorithms in general, are highly parallel, and can be sped up superpolynomially in a parallel computer. 
\begin{theorem}
The Graph Isomorphism problem can be solved using a parallel polynomial time algorithm with a quasipolynomial number of processors.
\end{theorem}

Later in this paper, we will show that Babai's work, as well as GI algorithms in general, uses a backtracking tree structure. The tree structure enables simple parallelization scheme. However, in Babai's algorithm the computation of each node in the tree is superpolynomial, and therein lies challenge of parallelizing the tree nodes. Several operations in the computation of each node that are bottlenecks and we will develop a parallelization scheme for each of them. One particular nontrivial parallelization is the parallelization of the $k$-dimensional WL refinement procedure, which is a highly sequential procedure where each iteration depends on the result of the previous iteration.

The structure of the paper is as follows. Section 2 presents preliminary techniques and a short overview of Babai's algorithm. Section 3 identifies the bottlenecks in Babai's algorithm and parallelization schemes designed to handle them. Section 4 presents our conclusion and the impact of this result.

\section{Preliminaries}
\subsection{Combinatorial techniques: Color refinement and $k$ dimensional WL refinement, individualization}
\label{subsec: combinatorial}
Color refinement has been used in most of the existing practical or theoretical GI algorithms. The idea behind color refinement is to use colors to represent similarities. Let us use $C(v)$ to denote the color of a vertex $v$. A vertex $v$ from graph $G_1$ with color $C(v)$ can only be mapped to a vertex $u$ in $G_2$ with color $C(u)$ if $C(u) \equiv C(v)$. Similarity can be evaluated using multiple criteria, for example the degree of a vertex. A common color refinement scheme is to evaluate the similarity of two vertices based on the colors of their neighbors, i.e. $C(u) \ne C(v)$ if $\{C(u')|u' \in neighbors(u)\} \ne \{C(v') | v' \in neighbors(v)\}$. After comparing the vertices and updating the new colors, the process repeats until equilibrium is achieved - no additional color class is created. Two graphs are isomorphic only if they have the same coloring after refinement. Note that we call the coloring, as well as any partition scheme, \textit{canonical} if they must be preserved by isomorphism and therefore must be agreed between two graphs.

Color refinement is a powerful tool as the GI problem for random graphs can usually be solved using one iteration of color refinement \cite{Babai2}. However, color refinement fails to distinguish graphs in hard cases, a notable example being the case of regular graphs. Therefore, an extension of color refinement has been developed by Weisfeiler-Lehman \cite{WL}: WL refinement, and then later $k$-dimensional WL refinement. Note that color refinement is usually referred to as 1-dimensional WL refinement.

The classic WL refinement is an extension of color refinement. For a graph $G = (V,E)$ and two elements $x,y  \in V$, the new color $C(x,y)$ encodes the old color and, for all $j,k$, the number of elements $z \in \Omega$ such that $C(x,z) = j$ and $C(y,z) = k$. Thus, the classic WL refinement refines the colors of not only the vertices but also the edges. From construction, the classic WL refinement is also called the 2-WL refinement.

The $k$-dimensional WL refinement further leverages the classic technique to work on tuples of size $k$ of the vertices. For a $k$-tuple $(x_1,\ldots x_k)$ of vertices, the new color encodes the old color and for all tuples $\{j_1, j_2 , \ldots j_k\}$, the number of elements $y$ such that $(\forall i < k)$ $C(x_1, \ldots, x_{i-1}, y, x_{i+1}, \ldots, x_k) = j_i$. It is worth noting that the classic WL refinement and the $k$-dimensional WL refinement can be extended to work with \textit{relational structures}, a generalization of graphs, as in Babai's paper \cite{Babai1}.

On the other hand, \textit{individualization} is a technique that can complement color refinement. Given two graphs $G_1$ and $G_2$, individualization maps a fixed vertex $v$ from $G_1$ to a vertex in $G_2$, usually by giving them a unique color. Since multiple vertices in $G_2$ could potentially be mapped to $v$ (i.e. have the same color as $v$), a GI algorithm that uses individualization has to consider all possible mappings. Therefore, individualization generates a multiplicative cost to the runtime of the program. Combining individualization and color refinement will produce a program with a backtracking search tree structure, see McKay \cite{McKay1}.

\subsection{Group theoretic techniques: Group procedures}
Group theory appears naturally in the course of solving Graph Isomorphism, because the set of automorphisms of a graph $G$ naturally makes up a group, and the set of isomorphisms between two graphs $G$ and $G'$ must be a coset of that group. 

%For example, color refinement can be seen as a procedure to restrict the possible automorphism group of a graph to a subgroup which consists of the tensor product of the automorphism groups of each color class.

Here we present some basic permutation group concepts and notations. A permutation group is a group that contains the permutations of a set of objects. The set of objects is called the permutation domain, or domain set, and we say that the permutation group acts on that domain. The symmetric group of a domain $\Omega$ is the group containing all permutations of that domain, denoting $S_\Omega$ or $S_n$ if $|\Omega| = n$. Given a permutation group $\mathcal{G}$ subgroup of $S_\Omega$ (denotes $\mathcal{G} \le S_\Omega$), a $\mathcal{G}$ orbit of a set element $a \in \Omega$ is the set of elements to which $a$ can be mapped through the group action. Orbits therefore are equivalence classes of objects that can be mapped to each other. If a group has only one orbit, then the group is called transitive i.e. every pair of objects can be mapped to each other. If a group acts transitively on a domain set, a block is a subset that either gets mapped to itself, or gets translated somewhere else entirely. The block system naturally forms a partition of the domain set, and the minimal block system is the block system with the minimum number of blocks (the block size is maximum).

In GI computation, a group is usually represented using its set of generators. We will list some useful computation group operations that appear in the literature. Given a group $\mathcal{G} \le S_n$) with generator set $A$, we can
\begin{itemize}
    \item find the orbits of $\mathcal{G}$ and transitivity of $\mathcal{G}$ by iterating through the images of the domain set elements.
    \item find $\mathcal{G}$'s minimal block system by following \cite{Helfgott}, who quotes \cite{Luks:1980} and \cite{Sims1}: fix $a \in \Omega$, for $b \ne a \in \Omega$, construct the graph with $\Omega$ its set of vertices and $\{\{a,b\}^g,g \in \mathcal{G}\}$ as edges, then the connected component containing $a$ and $b$ is the smallest block containing $a$ and $b$. The action of $\mathcal{G}$ is imprimitive if and only if the constructed graph is not connected for an arbitrary $a$ and at least one $b$, thus we obtain the block containing $\{a,b\}$ thus the block system must be non-trivial.
    \item generate the "whole" group given a set of generators by using the Furst-Hopcroft-Luks (FHL) algorithm \cite{Luks:1980}. From FHL, we can also determine whether an element belongs to a group, and determine the subgroup of a group given a polynomial time membership testing of the subgroup.
\end{itemize}
It is important to note that all these operations run in polynomial time in $n$ and $|A|$.
\subsection{Parallel computing and the PRAM model}
% Parallel computing is a type of computing where several processors execute simultaneously. The runtime of a parallel algorithm is the time that elapses from when the algorithm starts to the moment the last processor finishes execution.

In this work, we build our parallel algorithm using the Parallel RAM (PRAM) computational model. Unlike the typical Random Access Memory (RAM) model where there is one processor and instructions are executed sequentially, the PRAM model features multiple processors, local memory of each processor and global shared memory, and free read and write accesses to the global shared memory. 

\subsection{Babai's quasipolynomial time Graph Isomorphism algorithm}

In 2015, Babai showed that the String Isomorphism problem, a more general extension of GI, can be solved in quasipolynomial time. The String Isomorphism (SI) problem, given two string inputs $x$ and $y$ as functions from indices set $\Omega$ to symbols $\Sigma$ and a subgroup $\mathcal{G} \le S_\Omega$, determines whether or not there exists at least one element $\tau \in \mathcal{G}$ such that $\tau$ is an isomorphism between $x$ and $y$, i.e. $\forall i \in \Omega$ $x(i) = y(\tau(i))$. The Graph Isomorphism problem can be reduced to the String Isomorphism problem by flattening the input graphs' adjacency matrices into the strings and constructing $\mathcal{G}$ corresponding to this transformation.

\begin{theorem}
[Babai] The String Isomorphism problem can be solved in quasipolynomial time. As a result, the Graph Isomorphism problem can be solved in quasipolynomial time.
\end{theorem}

In the interest of brevity, we will go through the novel concepts that make the algorithm succeed and will give a high level description of the algorithm. %Some improvements made by Helfgott \cite{Helfgott} will be incorporated to make the algorithm simpler to explain and easier to reason about. 
We refer the readers to the work of Babai \cite{Babai1} and Helfgott's explanatory document \cite{Helfgott} to know more about the details of the algorithm.

Babai's algorithm includes several key features that combine and complement each other beautifully to make the algorithm succeed:
\begin{itemize}
    \item Divide and conquer, and recursion: at almost any point in the algorithm, the high-level goal is to either recursively reduce the underlying group into a combination of its subgroup and cosets, or the domain set into smaller subsets, and consider smaller problems having those smaller subgroups or subsets as input. This continues until the input sizes become sufficiently small such that each problem instance can be solved in polynomial time using brute force. The algorithm runs in quasipolynomial time if at any recursive step the problem can be broken down into a quasipolynomial number of smaller instances, and the depth of recursion is at most polylogarithmic.
    \item Group theoretic and combinatorial techniques: Both combinatorial techniques and group theoretic techniques are used in Babai's algorithm. Intuitively, combinatorial techniques, namely partitioning techniques, work in a top-down fashion - those techniques try to partition the graph based on high level asymmetry. In contrast, group theoretic techniques approach the problem bottom up and try to construct the automorphism group element by element. Generally, combinatorial techniques work well when the partitioning structure is highly asymmetric, whereas group theoretical technique work well when the partitioning structure is highly symmetric.
    \item Symmetry vs asymmetry: Both symmetry and asymmetry benefit the algorithm. Asymmetry enables combinatorial techniques to reduce the problem to smaller instances as stated above. On the other hand, symmetry enables canonical splitting of the domain set into the symmetric part and the asymmetric part, where the automorphism group of the symmetric part is easy to construct due to its structure.
    %TODO: 
\end{itemize}  

There are five main procedures in Babai's algorithm: \textbf{Luks' reduction}, \textbf{Local Certificate}, \textbf{Aggregate Certificates}, \textbf{Design Lemma} and \textbf{Split-or-Johnson}. Luks' reduction is the first step in each recursion, where Babai employs the framework developed by Luks \cite{Luks:1980} and reduces the bottleneck of GI to a certain type of group, whose structure enables the transformation of $\Omega$ to an auxiliary domain $\Gamma$. Local Certificate examines a logarithmic-size subset of $\Gamma$ to determine whether it is highly symmetric with respect to the current underlying group. Aggregate Certificates combines the results of Local Certificates of all those subsets to infer properties about the global symmetry or asymmetry. If there are a lot of symmetric test sets, then the global symmetries allow efficient recursion. Otherwise, if there are a lot of asymmetric test sets, the graph must be highly asymmetric. Design Lemma takes a highly asymmetric relational structure as the input and produces a uniprimitive coherent configuration, which is an important algebraic structure closely related to the 2-dimensional WL refinement, %TODO: add ref
in the worst case. Finally, Split-or-Johnson takes a uniprimitive coherent configuration and produces either a "split" i.e. a canonical partition of the domain set, or a large canonically embedded graph whose structure is well-known and thus enables efficient recursion.

\section{Parallel Graph Isomorphism algorithm}
\subsection{Identifying the bottlenecks}
In this section, we identify all the superpolynomial bottlenecks in Babai's algorithm. We refer the reader to the technical details of Babai's paper \cite{Babai1} or to Helfgott's detailed time complexity analysis of Babai's algorithm \cite{Helfgott} to verify that the bottlenecks mentioned are indeed all superpolynomial bottlenecks in the algorithm. We will also identify the bottlenecks that are easy to solve, i.e. embarrassingly parallel ones.

In Babai's algorithm, there are two main types of bottleneck: multiplicative bottlenecks and large computation bottlenecks.

\subsubsection{Multiplicative bottlenecks}
Multiplicative bottlenecks are bottlenecks that arise from either individualization steps or recursion steps. These steps all create change to the structure of the program similar to branching out from a tree, and therefore introduce a multiplicative cost to the program.

Individualization is a step that occurs in the Design Lemma procedure, the Split-or-Johnson procedure and the Aggregate Certificates procedure. In fact, individualization is the main action of the Design Lemma procedure. As stated in subsection \ref{subsec: combinatorial}, individualization induces a multiplicative cost to the program's runtime.

Recursion or reduction of a problem instance to multiple smaller instances also cause a multiplicative cost to the program's runtime. The reduction steps create smaller problem instances with no guarantee about runtime (we can only reason about the maximum depth of recursion). The steps occur in the reduction of the working group to its subgroup and the cosets, the reduction to a Johnson group in Luks framework, and in the recursive call to a quasipolynomial number of instances of SI on smaller inputs in the Local Certificate procedure.

Due to their tree-like structure, the computations of the multiplicative bottlenecks are embarrassingly parallel.

\begin{lemma}
\label{lemma: 1}
Multiplicative bottlenecks can be computed in parallel in polynomial time,
such that any multiplicative cost to the runtime of the programs becomes multiplicative
cost to the number of processors. Furthermore, the parallelization is work-preserving.
\end{lemma}

This case is straightforward. For the recursion/reduction steps, we assign a processor for each problem instance created. Similarly for the individualization steps, we assign a processor for each of the vertex choice possibilities. Since the recursion/reduction problem instances as well as the possible cases of individualization are independent, they can be solved in parallel independently. Therefore, their multiplicative cost to the runtime of the program becomes multiplicative cost to the number of processors needed.

\subsubsection{Large computation bottlenecks}
Large computation bottlenecks are the superpolynomial tasks which appear throughout Babai's algorithm. They are: 1) the large repetitive tasks, 2) the aggregating results and procedures on large groups tasks, and 3) the $k$-dimensional WL refinement task.

Large repetitive tasks are computational tasks that consist of a process that repeats itself multiple times independently. There are two such tasks in Babai's algorithm: one is the repeating of Local Certificate procedure for $\binom{|\Gamma|}{k}$ test sets where $k = \log{}n$ at the start of Aggregating Certificates, and the other is the search through possibly all tuples of $\le k$ elements to see if individualizing those elements yields the desired outcome. Again, due to each iteration being independent, those tasks are embarrassingly parallel.

\begin{lemma}
\label{lemma: 2}
Large repetitive tasks can be computed in parallel in polynomial time. Furthermore, the parallelization is work-preserving.
\end{lemma}

Aggregating results tasks arise when the algorithm needs to combine results from multiple smaller problem instances or subprocedures. This type of task occurs in multiple places in Babai's algorithm. For example, almost all reduction steps, such as the reduction step in Luks framework or the reduction step in the Local Certificate procedure, require combining the resulting groups or cosets (represented by their generating sets) into one group by merging their generating sets. The Aggregating Certificates procedure also combines the certificates into one group in a similar fashion. This is not problematic memory-wise, since we can avoid memory complications between processors by using the PRAM model of computation. However, this creates large group (in term of size of generating sets). As we mentioned above, many group operations have runtime polynomial in the size of the generating set. Therefore, this causes a superpolynomial bottleneck. In the next part of the paper we develop a generating set refinement procedure to deal exclusively with this type of bottleneck.

The last type of bottleneck is the $k$-dimensional WL refinement. This operation deserves a separate subsection, since it is a highly sequential and highly non-trivial obstacle to parallelization. Our approach to parallelize this operation will be discussed in detail in subsection \ref{subsec: parallel WL}.

\subsection{Group generators refinement}
\label{sec: refine generating}
As stated previously, group operations usually run in polynomial time of the size of the generating set. Therefore, operations on a group with a large (non-polynomial) generating set is problematic. 

In this subsection, we develop a subroutine that refines the generating set of any permutation group. 

\begin{theorem}
Given a permutation group $G \le S_n$ with generating set $A$, we can create a new generating set $A'$ of $G$ which has size $O(n\log{}n)$. The process can be done in parallel in $O(n\log{}n)$ steps using $|A|$ number of processors.
\end{theorem}
\begin{proof}
The subroutine is based on the FHL algorithm (also called Schreier-Sims algorithm), an algorithm described by Luks in \cite{Luks:1980}. Given $G$ some permutation group with domain set $\Omega = \{x_1, x_2, \ldots x_n \}$, the FHL algorithm works with a tower of stabilizers $G = G_0 > G_1 > G_2 > \ldots > G_{n-1} = {e}$, where $G_i$ is the group of permutations that keeps all elements $x_1, x_2 \ldots x_i$ fixed. The algorithm builds a set of $C_i$ of representatives of $G_i/G_{i-1}$ and thus $\bigcup_{i \le j < n-1} C_j$ generates $G_i$ for all $i < n$. 

We note that the FHL algorithm itself runs in polynomial time of the size of the group's generating set. Our subroutine is a derivation of the FHL algorithm, but is parallel so that it can handle input groups with large generating sets.

\begin{algorithm}[H]
\caption{MembershipTest}
\textbf{Input}: permutation $x$, set of coset representatives $C_0, C_1, \ldots C_{n-2}$ \\
\textbf{Output}: true/false, index $j$
\begin{enumerate}
    \item $j = 0$
    \item while $j < n-1$ do:
    \begin{enumerate}
        \item \textbf{if} $\gamma^{-1}x' \in G_{i+1}$ for some $\gamma \in C_i$ \textbf{then} $x' \leftarrow \gamma^{-1}x'$
        \item \textbf{else} \textbf{return} (false, $j$)
    \end{enumerate}
    \item \textbf{return} (true, -1)

\end{enumerate}
\end{algorithm}
\begin{algorithm}[H]
\caption{RefineGeneratingSet}
\textbf{Input}: $A = \{x_1, x_2, \ldots x_N\}$ a generating set of a group $G \le S_n$ \\
\textbf{Output}: $A' = \{x'_1, x'_2,\ldots x'_{N'}\}$ a generating set of $G$ where $N' = O(n\log{}n)$
\begin{enumerate}
    \item $C_i \leftarrow e$ for $i = 0$ to $n-2$
    \item $A' \leftarrow \{e\}$ 
    \item \textbf{while} true:
    \begin{enumerate}
        \item \textbf{parfor} all elements $x \in A$ \textbf{do}: 
        \begin{enumerate}
            \item call MembershipTest($x$, $C_1, \ldots C_{n-2}$)
        \end{enumerate}
        \item \textbf{if} all element return true, \textbf{break}
        \item $x \leftarrow$ a randomly chosen element from the ones that return false
        \item $A' \leftarrow \{A' \cup x\}$, $C_j = \{C_j \cup x\}$ for $j$ = the index returned by MembershipTest($x$, $C_1, \ldots C_{n-2}$)
    \end{enumerate}
    \item return $A'$
\end{enumerate}
\end{algorithm}

Proof of correctness: Denote $G(A')$ the group generated by generating set $A'$. The membership testing method is derived from FHL algorithm: for each element $x$, for increasing $i$ the algorithm searches for a representative of the coset of $x$ modulo $G_{i+1}$ in the set $C_i$. If such a representative is found, then $x$ is already contained in the group created by $A'$. The RefineGeneratingSet procedure builds $A'$ by gradually adding elements from $A$ to it gradually. In each iteration, elements of $A$ are tested to see if they are a member of  $G(A')$. It chooses one such element and updates $A'$ and the $C_i$s. The algorithm stops when there are no more elements in $A$ can be added to $A'$, in other word all elements in $A$ are contained $G(A')$.

Proof of time complexity: First, it is easy to see that the membership testing function runs in polynomial time of $n$ and $A'$ (note that $\bigcup_i C_i = A'$). Thus, the \textit{parfor} loop runs in polynomial time using $|A|$ processors. Therefore, the runtime of the algorithm is polynomial if the number of iterations in the main while loop (line 3), which is $|A'|$, is polynomial. In each of its iterations, the main while loop finds one element that is not contained in $G(A')$ and add it to $A'$. Doing this results in a $G(A')_{new}$ that contains $G(A')_{old}$ as a proper subgroup, meaning $G(A')_{new}$ contains $G(A')_{old}$ and at least one coset. Therefore, $|G(A')_{new}| \ge 2|G(A')_{old}|$ and thus $|G(A')| \ge 2^k$ where $k = |A'|$. We note that $G(A')$ is always a subset of the symmetric set, thus $|G(A')| \le n!$. Therefore, $2^k \le n!$ which means $k \le \log{}(n!) < n\log{}n$. 

Therefore, the newly created generating set $A'$ has size polynomial and the subroutine can be done in parallel polynomial time using $|A|$ processors.
\end{proof}

\subsection{Parallelization of $k$-dimensional WL refinement} 
We here introduce one of the main contributions of this paper: the parallelization of the $k$-dimensional WL refinement. 

In Babai's algorithm, the $k$-dimensional WL refinement procedure appears in the Split-or-Johnson procedure. Recall the definition of $k$-dimensional WL refinement in subsection \ref{subsec: combinatorial}. Similar to color refinement, $k$-dimensional WL refinement is a iterative procedure where the next iteration depends on the results of previous iterations. No proven upperbound on the number of iterations of $k$-dimensional WL refinement exists other than the trivial $n^k$ one (each iteration  procedures at least one new color and there are at most $n^k$ colors). Therefore, the $k$-dimensional WL refinement procedure is a highly non-trivial obstacle for parallelization.

In this subsection, we will derive a parallelization scheme for the $k$-dimensional $WL$ refinement. We will first develop a parallelization scheme for the 1-dimensional WL refinement by leveraging it to 2-dimensional WL refinement, and then develop a transformation between performing $k$-dimensional WL refinement and performing 1-dimensional WL refinement. 

\subsubsection{Color refinement and 2-dimensional WL refinement}
\label{sec: CR WL}
Let us first consider 1-dimensional WL refinement i.e. color refinement. The number of iterations of color refinement has a tight upperbound $O(n)$, as there are graphs where color refinement requires a linear number of iterations such as the line graph. 

We will now take a different perspective by looking at a pair of vertices distinguished in the $h^{th}$ iteration of color refinement. Consider a graph $G = (V,E)$ with initial coloring $\mathcal{C}$ of the vertices. We can expand the initial coloring to 2-tuples by assigning color of $\mathcal{C}(u,u)$ by the color of $u$ and $\mathcal{C}(u,v)$ based on whether or not $(u,v) \in E$. Then, we employ the idea of \textit{walk refinement} by Lichter \cite{Lichter}. Given a $G = (V, E)$ with an initial edge coloring $C$, the $k$ walk refinement is a refinement scheme that for $G$ and $C$ determines a new coloring $\mathcal{C}_{W[k]}$ defined by: 
\begin{equation*}
    \mathcal{C}_{W[k]}(u,v) =  \{ \{ {\mathcal{C}}(u, w_{i_1}, w_{i_2}, \ldots, w_{i_{k-1}}, v) \}, w_{i_j} \in V \}
\end{equation*} 
where ${\mathcal{C}}(v_1, v_2 ,\ldots v_k) = (C(v_1,v_2), C(v_2, v_3) \ldots C(v_{k-1}, v_k))$.

Intuitively, a $k$-walk refinement refines the color of a pair $(u,v)$ by comparing the colors of all the possible walks of length $h$ from $u$ to $v$. Note that a $k$-walk refinement also implicitly contains $k'$-walks by having $k-k'$ first steps stationary.

\begin{theorem}Given a graph $G = (V,E)$ with two vertices $u$ and $v$ in $V$ distinguished in the $h^{th}$ iteration of color refinement, $h >1$. 
Then $u,v$ can be distinguished by a $2h$-walk refinement, i.e. $\mathcal{C}_{W[2h]}(u,u) \ne \mathcal{C}_{W[2h]}(v,v)$. 

% Then there must be a pair of paths of length $\le h$ from $u$ to $\bar{u}$ and $v$ to $\bar{v}$, where $\bar{u}, \bar{v}$ are vertices in $V$ and $\bar{u}, \bar{v}$ are distinguished in the first iteration. Furthermore, these conditions must hold:
% \begin{romanenumerate}
% \item these paths must be the shortest from $u$ to $\bar{u}$ and $v$ to $\bar{v}$,
% \item every pair of vertices of the same distance $l$ from $\bar{u}$ and $\bar{v}$ in the 2 paths are distinguished in the respective $(l+1)^{th}$ iteration   , and
% \item there are no pair of vertices $\bar{u}', \bar{v}'$ distinguished in the first iteration having shorter paths to $u, v$.
% \end{romanenumerate} 
\end{theorem}
\begin{proof}
Proof by induction. The statement is trivial true for $h = 1$ where the algorithm simply returns the initial colors.

For $u, v$ a pair of vertices distinguished in the $h^{th}$ iterations, they must satisfy these conditions:
\begin{romanenumerate}
\item the set of colors of the neighbors of $u$ = $\{u_1, u_2 \ldots u_j \}$ and the set of colors of the neighbors of $v$ = $\{v_1, v_2,\ldots v_j \}$ (assuming they have the same number of neighbors) must be different in the $h^{th}$ iteration.
\item In every iteration $h' < h-1$, $u$ and $v$ must have exactly the same neighbor color set.
\end{romanenumerate}

Consider a pair of vertices $u'$ and $v'$ in the neighbors of $u$ and $v$ respectively distinguished by the $(h-1)^{th}$ iteration of color refinement. By induction hypothesis, $u'$ and $v'$ are also distinguished by the $2(h-2)$ walk refinement, i.e. $\mathcal{C}_{W[2h-2]}(u',u') \ne \mathcal{C}_{W[2h-2]}(v',v')$.

Now, consider $\mathcal{C}_{W[2h]}(u,u)$ and $\mathcal{C}_{W[2h]}(v,v)$. These colors are distinguished using $\{{\mathcal{C}}(u, w_1$ $, \ldots w_{2h-1}, u \}$ and $\{{\mathcal{C}}(v, w_1, \ldots w_{2h-1}, v\}$ . We will only consider a subset of all the walks which is $\{u, u_i, w_1, \ldots w_{2k-3}, u_i, u \}$ and $\{v, v_i, w_1, \ldots w_{2k-3}, v_i, v \}$ for $i = 1,2 \ldots j$ and $u_1, u_2,$ $\ldots u_j$ are the neighbors of $u$, and similar for $v$. We will prove that these two subsets of walks result in a different set of colors, and that these two subsets having different colors result in $\mathcal{C}_{W[2h]}(u,u) \ne \mathcal{C}_{W[2h]}(v,v)$.

First, since $u$ and $v$ are distinguished in the $h^{th}$ iteration of color refinement, their set of neighbor colors must be different after the $(h-1)^{th}$ color refinement iteration. From the induction hypothesis, this means that $\{\{\mathcal{C}_{W[2h-2]}(u_i,u_i)\}\} \ne \{\{\mathcal{C}_{W[2h-2]}(v_i,v_i)\}\}$ where the $u_i, v_i$ are neighbors of $u,v$ respectively. This means that $\{ \{ {\mathcal{C}}(u_i, w_1, \ldots, w_{2h-3}, u_i) \} \} \ne \{\{{\mathcal{C}}(v_i, w_1, \ldots w_{2h-3}, v_i)$. Therefore, ${\mathcal{C}}\{(u, u_i, w_1, \ldots w_{2k-3}, u_i, u) \}$ $\ne$ ${\mathcal{C}}\{(v, v_i$ $, w_1, \ldots w_{2k-3},$ $v_i, v) \}$, thus those two subsets of walks result in a different set of colors.

Second, we will prove that these subsets of walks have a unique role in all the $2h$ walks considered. This is indeed true, because $\mathcal{C}(u,u_i)$ the initial color of an edge is unique from the definition.

Thus, if $u$ and $v$ are distinguished in the $h^{th}$ iteration of color refinement, then they can be distinguished by a $2h$-walk refinement.
\end{proof}

However, Lichter \cite{Lichter} proved that a $2h$-walk refinement can be simulated by $\log{}(2h)$ iterations of 2-dimensional WL refinement. 

\begin{theorem} [Lichter] $k$-walk refinement can be simulated with $\left \lceil{\log{}k}\right \rceil $ iterations of Weisfeiler-Leman refinements.
\end{theorem}

From the above two theorems, we immediately come to our conclusion.
\begin{corollary} \label{cor: CR WL}
Given a graph $G = (V,E)$, if color refinement stabilizes in $h$ iterations, then we can use 2-dimensional WL refinement to simulate the result using $\log{}2h$ iterations.
\end{corollary}

We note that this result has significant implications beyond Babai's algorithm. 

\begin{theorem} 
There is a parallel algorithm to simulate the result of color refinement in logarithmic time.
\end{theorem}

\begin{proof}
We can use 2-dimensional WL refinement to simulate color refinement as stated above. Note that each step of the 2-dimensional WL refinement is completely parallelizable (i.e. can be done in constant time by a work-preserving parallel algorithm), therefore the runtime of the parallel algorithm is the number of iterations, which is $O(\log{}n)$ as from the above two theorems.
\end{proof}

Achieving the result of color refinement in logarithmic time is remarkably useful for practical GI solvers, which mainly use color refinement instead of the complicated $k$-dimensional version. Furthermore, this result implies that any decision problem which can be NC-reduced to color refinement is in NC, the class of problems that can be solved efficiently  in a parallel computer.

\subsubsection{$k$-dimensional WL refinement and color refinement} 
Now we will present a way to simulate $k$-dimensional WL refinement using color refinement

\begin{theorem} 
\label{theorem: WL final}
Given a graph $G = (V,E)$, there is a graph $G'$ such that performing $k$-dimensional WL refinement on $G'$ can be simulated by performing color refinement on $G'$.
\end{theorem}
\begin{proof}
Recall that the $k$-dimensional WL refinement determines the new color of a tuple $(w_1, w_2 \ldots w_k)$ based on the number of $\{ y \in V | (\forall i < k) (C(w_1,\ldots w_{i-1}, y, w_{i+1}, \ldots, w_k) = c_i)\}$ for all tuples of color $c_1, c_2 \ldots c_k$.

Denote $V$ $= \{w_1, w_2 ,\ldots w_n\}$. We construct $G'$ as follows:
For each $k$ tuple of vertices $(w_{i_1}, w_{i_2}, \ldots w_{i_k})$ in $G$, create a corresponding vertex $x$ in $G'$. Let us write $x \approx (w_{i_1}, w_{i_2} \ldots w_{i_k})$, and call all the $x$s the \textit{base layer}.
% For each of these vertices, construct the auxiliary $k$ neighbors $u_1, u_2 \ldots u_k$, each connected to $u$, and for each $i$ connect $u_i$ with the vertices corresponding to $(w_1, \ldots w_{i-1}, x, w_{i+1}, \ldots w_k)$. The initial color of each $k$-tuple vertices in $G'$ is similar to the initial color of the corresponding $k$-tuple in $G$, and all the auxiliary $u_i$ are given a unique same color.
For each vertex $x$ of $G'$ that corresponds to the tuple $(w_{i_1}, w_{i_2} \ldots w_{i_k})$, construct two layers of auxiliary nodes. The first layer consists of $n$  nodes $u_1, u_2, \ldots u_n$. Connect each of the $u_i$ with $x$. 
The second layer consists of $nk$ nodes $v_{1,1}, v_{1,2}, \ldots v_{1,k}, v_{2,1}, v_{2,2} \ldots v_{n,k}$. Connect each $v_{i,j}$ with $u_i$ and with the node that corresponds to the $k$-tuple $(w_{i_1}, w_{i_2}, \ldots, w_{i_{j-1}}, w_i, w_{i_{j+1}}, \ldots w_{i_k})$. We say that these nodes in the auxiliary $u$ and $v$ layer \textit{come out} from $x$ in order to distinguish them with the auxiliary nodes that expand from other vertices, or \textit{come in}, to $x$, and call the layer of $u$ nodes that expand from $x$ as constructed above the \textit{out-layer} and the layer of $y$ nodes that connect to $x$ the \textit{in-layer}.  The initial color of each $x$ in $G'$ is similar to the initial color of the corresponding $k$-tuple in $G$, all the auxiliary $u_i$ are given a same new color, and for each $j$ all the auxiliary $v_{i,j}$ are given a same new color. The construction scheme of the out-layers of node $x$ in $G'$ corresponding to a $k$-tuple $(w_{i_1}, w_{i_2} \ldots w_{i_k})$ is as illustrated in the following figure. 

\begin{figure}[H]
    \centering
    \includegraphics[width = 0.8\textwidth]{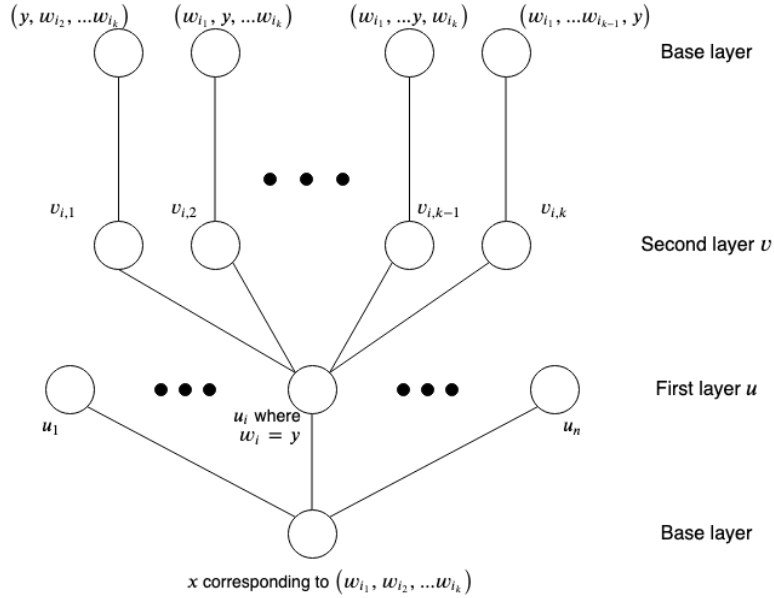}
    \caption{Illustration of constructing the auxiliary out-layers from a vertex $x$ in $G'$ to simulate performing $k$-dimensional WL refinement on $G$}
    \label{fig:expansion}
\end{figure}

For vertex $x \in G'$ corresponding to the tuple $(w_{i_1}, w_{i_2} \ldots w_{i_k})$ of vertices in $G$, denote $\overrightarrow{W}_j^y$ as the tuple $(w_{i_1}, \ldots w_{i_{j-1}}, y, w_{i_{j+1}}, \ldots w_k)$ and $\overrightarrow{x}_j^y$ the corresponding vertex. Intuitively, for each $x \approx (w_1, w_2 \ldots w_k)$ the auxiliary $u_i$ is used to capture the set of tuples $\overrightarrow{W}_j^{w_i}$ for all $j \le n$, and the auxiliary $v_{i,j}$ with distinct initial colors for different $j$ are used to make sure that the exact ordering $j$ of the sequence of colors of $\overrightarrow{W}_j^{w_i}$ affects the colors of $u_i$. Note that the different layers are always distinguished since their initial colors are different.

% We will prove that each iteration of $k$-dimensional WL refinement on $G$ can be simulated by two iterations of color refinement on $G'$. Proof by induction. For the base case 
We will prove that color refinement on $G'$ results in the same coloring of the base layer nodes as performing $k$-dimensional WL refinement on $G$.

Consider 2 $k$-tuples $W = (w_{i_1}, w_{i_2} \ldots w_{i_k})$ and $W' = (w_{i'_1}, w_{i'_2} \ldots w_{i'_k})$ of vertices in $G$ that are not distinguished at the end of $k$-dimensional WL refinement. In this case, the tuples' corresponding vertices $x$ and $x'$ in $G'$ must have the same out-layer. Indeed, each $i$ the $u_i$ nodes captures the color determined by the set of colors of $\overrightarrow{x}_j^{w_i}$. Since $W$ and $W'$ are not distinguished by the $k$-dimensional WL refinement, the set of $\overrightarrow{x}_j^{w_i}$ and $\overrightarrow{x'}_j^{w_i}$ must be symmetric. Therefore, the $u$ nodes of $x$ and $x'$ must be identical up to permutation. On the other hand, $x$ and $x'$ must also have the same in-layer, because the color of each of the $v$ nodes coming from $\overrightarrow{x}_j^{y}$ depends on the color of $x$ and of $\overrightarrow{x}_j^{y}$, which must be symmetric between $x$ and $x'$. 

Now, if $W$ and $W'$ are distinguished at the end of $k$-dimensional WL refinement, then it is straightforward to see that $x$ and $x'$ are distinguished by the end of color refinement on $G'$. This is because the out-layer $v_{i,j}$ of each $u_i$ node $x$ captures the colors of $\overrightarrow{x}_j^{w_i}$ not preserved under permutation due to the initial coloring of $v_{i,j}$. Therefore, asymmetries in the colors of $\overrightarrow{x}_j^{w_i}$ and $\overrightarrow{x'}_j^{w_i}$ (which are detected at some point because $W$ and $W'$ are distinguished) lead to asymmetries in the colors of the auxiliary $u_i$ of $x$ and $x'$, which in turn leads to different colors of $x$ and $x'$.

% Else if the $k$-dimensional WL refinement has not distinguish tuples $(w_1, w_2 \ldots w_k)$ and $(w'_1, w'_2 \ldots w'_k)$ by iteration $h$, then color refinement on $G'$ has not distinguish the corresponding pair of vertices in $G'$ as well as the associated set of auxiliary nodes of the pair of vertices by iteration $3h$.

\end{proof}

\subsubsection{Parallelizing $k$-dimensional WL refinement}
\label{subsec: parallel WL}
We arrive at our main goal of this subsection.
\begin{theorem} \label{theorem: k dim}
Given a graph $G = (V,E)$, we can perform $k$-dimensional WL refinement, where $k = O(\log{}n)$, in parallel in polynomial time using a quasipolynomial number of processors. 
\end{theorem}
\begin{proof}
First, we use the simulation in theorem \ref{theorem: WL final} to convert performing $k$-dimensional WL refinement on $G$ to performing color refinement on $G'$. From construction, $G'$ has $O(nkn^k)$ vertices, therefore color refinement on $G'$ runs at most $O(nkn^k)$ iterations. 

Next, we use 2-dimensional WL refinement on $G'$ to simulate color refinement on $G'$ as in Corollary \ref{cor: CR WL}. The result of color refinement can be simulated in $\log{(2nkn^k)}$, which is $ O(n)$ iterations of 2-dimensional WL refinement for $k = O(\log{}n)$. 

Finally, we can execute this 2-dimensional WL refinement in parallel polynomial time. We will prove that in each iteration of 2-dimensional WL refinement, the computation can be done efficiently in parallel. Recall that in the 2-dimensional WL refinement, the new color of a pair of vertices $u,v$ is determined by the number of elements $z \in \Omega$ such that $c(x,z) = j$ and $c(y,z) = k$ for all $j,k$. For each pair of vertices $(u,v)$, looking at the colors $c(u,j)$ and $c(j,v)$ for all vertices $j$ in $G'$ is entirely parallelizable and can be done in constant time. Comparing those sets of color between two pair of vertices takes constant time, and performing all the comparisons among all sets of two pairs of vertices can be done in parallel. Therefore, each iteration of 2-dimensional WL refinement of $G'$ can be done in parallel polynomial time, thus whether the time complexity of a parallel execution of 2-dimensional WL refinement is polynomial or not depends on the total number of iterations needed, which is polynomial in $n$. Therefore, the time complexity of simulating color refinement on $G'$ using 2-dimensional WL refinement in parallel is polynomial, with a quasipolynomial total amount of work (because our parallelization of 2-dimensional WL refinement iterations is work-preserving, and there is a quasipolynomial number of vertices in $G'$).

In conclusion, we can perform $k$-dimensional WL refinement where $k = O(\log{}n)$ in parallel polynomial time with a quasipolynomial amount of total work, and therefore a quasipolynomial number of processors needed.
\end{proof}

\subsection{The combined algorithm}

We have examined all the bottlenecks in Babai's algorithm and parallelized them. In this subsection, we combine the parallelization schemes from the previous parts to create the master algorithm. 

\begin{theorem}
The graph isomorphism problem can be solved by a parallel algorithm in polynomial time using a quasipolynomial number of processors.
\end{theorem}
\begin{proof}
From lemma \ref{lemma: 1} and \ref{lemma: 2}, the multiplicative bottleneck and the large repetitive tasks bottleneck of Babai's algorithm can be done in parallel polynomial time in a work-preserving manner and such that any multiplicative cost to the runtime of the program becomes multiplicative cost to the number of processors.

For group operations on large groups, which are formed by aggregating results tasks, we apply our RefineGeneratingSet subroutine developed in subsection \ref{sec: refine generating} after every aggregation of results operation. By doing this, we guarantee that all the groups in the algorithm have a polynomial size generating set, and therefore the runtime of group operations on them is not superpolynomial. The parallel subroutine runs in  polynomial time and introduces an additional worst case quasipolynomial $n^{\log{}n}$ cost to the total amount of work.

For $k$-dimensional WL refinement, we follow the procedure described in theorem \ref{theorem: WL final} to parallelize the procedure to parallel polynomial time. The parallelization introduces a $n^{O(\log{}n)}$ multiplicative cost to the total amount of work. However, we note that this is an one-time cost, and thus the total amount of work done by the algorithm after parallelization is still quasipolynomial.

Combining the above arguments, we arrive at our parallel version of Babai's algorithm. The algorithm runs in polynomial time in each iteration, any multiplicative cost to the runtime of the algorithm after each iteration is converted into multiplicative cost to the number of parallel processors, and there are a polynomial (in fact logarithmic) number of iterations in total. Thus, the whole parallel algorithm runs in polynomial time. The amount of total work is increased by at most a quasipolynomial multiplicative cost, and thus is still quasipolynomial. Therefore, the number of processors needed is quasipolynomial. 

\end{proof}

\section{Conclusion and remarks}
In our work, we have proven that the GI problem can be solved using a parallel algorithm that runs in polynomial time using a quasipolynomial number of processors. Therefore, problems in the GI complexity class, which are problems that can be polynomially reduced to Graph Isomorphism, can be solved in parallel polynomial time using a quasipolynomial number of processors. 

This result implies that in theory solving the worst cases of the GI problem is tractable in a parallel computer. It also implies that Babai's algorithm is highly parallel and can be sped up superpolynomially in a parallel computer. We note that despite being quite complicated, Babai's algorithm shares multiple similarities with other GI algorithms. Many of those similarities are bottlenecks of parallelization, such as the backtracking search tree structure, color refinement and $k$-dimensional operations, individualization and so on. Therefore, the parallelization techniques used in this paper can be generalized to other GI algorithms as well, indicating that algorithms for GI in general is highly parallel.

In addition, our finding that color refinement can be calculated in logarithmic time in parallel can potentially have a notable impact in practical GI solvers. On the one hand, almost all practical GI solvers use color refinement as the main tool to solve the problem. On the other hand, given the increase in demand of big data processing and parallel algorithms, there have been a surprisingly few practical parallel algorithms for the GI problem - it is hard to find one such program beside \cite{Son15}. Therefore, our techniques used for parallelizing color refinement can potentially be applied to create new practical parallel GI algorithms, or to improve existing ones.

%%
%% Bibliography
%%

%% Please use bibtex, 

\bibliography{main}

\appendix

\end{document}